\documentclass[letterpaper, 10 pt, conference]{ieeeconf}

\bstctlcite{bstctl:etal, bstctl:nodash, bstctl:simpurl}

\usepackage{anyfontsize}
\usepackage{array} 
\usepackage[normalem]{ulem}
\usepackage{stfloats}   
\usepackage{mathtools,balance}
\usepackage{xcolor}
\usepackage[utf8]{inputenc}
\usepackage{tabularx}
\usepackage{booktabs} 
\usepackage{mathrsfs,stmaryrd}
\usepackage{graphics} 

\usepackage{amsmath} 
\usepackage{amssymb}  
\usepackage{amsthm}
\usepackage{cite}
\usepackage{bm}
\usepackage{acronym}
\usepackage{paralist}
\usepackage{float}
\usepackage[dvipsnames]{xcolor}
\usepackage{epstopdf}
\usepackage{multicol}
\usepackage{tikz}
\usepackage{hyperref}
\usepackage{graphicx}
\usepackage{soul}
\usepackage{macros} 

\usepackage{stfloats}
\usepackage{algorithm,algpseudocode,algorithmicx}
\fnbelowfloat

 \newcommand{\abby}[1]{{\color{red} [\underline{\bf AE:} #1]}}
  
\makeatletter
\hypersetup{colorlinks=true}
\AtBeginDocument{\@ifpackageloaded{hyperref}
  {\def\@linkcolor{blue}
  \def\@anchorcolor{red}
  \def\@citecolor{red}
  \def\@filecolor{red}
  \def\@urlcolor{black}
  \def\@menucolor{red}
  \def\@pagecolor{red}
\begingroup
  \@makeother\`%
  \@makeother\=%
  \edef\x{%
    \edef\noexpand\x{%
      \endgroup
      \noexpand\toks@{%
        \catcode 96=\noexpand\the\catcode`\noexpand\`\relax
        \catcode 61=\noexpand\the\catcode`\noexpand\=\relax
      }%
    }%
    \noexpand\x
  }%
\x
\@makeother\`
\@makeother\=
}{}}
\makeatother

\DeclareMathOperator{\R}{\mathbb R}

\newcommand{\bequ}{\begin{eqnarray}}
\newcommand{\eequ}{\end{eqnarray}}

\IEEEoverridecommandlockouts

\def\BibTeX{{\rm B\kern-.05em{\sc i\kern-.025em b}\kern-.08em
    T\kern-.1667em\lower.7ex\hbox{E}\kern-.125emX}}

\title{\LARGE \bf
{Optimal Dynamic Resource Allocation for Multicore Real-time Systems
}
}
\author{Georgiy A. Bondar, Abigail Eisenklam, Ripon C. Sarker, Linh Thi Xuan Phan, Abhishek Halder
\thanks{Georgiy A. Bondar is with the Department of Applied Mathematics, University of California Santa Cruz, Santa Cruz, CA 95064, USA, {\tt\footnotesize{gbondar@ucsc.edu}}. Abigail Eisenklam and Linh Thi Xuan Phan are with the Department of Computer and Information Science, University of Pennsylvania, PA 19104, USA, {\tt\footnotesize{\{aei,linhphan\}@cis.upenn.edu}}. Ripon C. Sarker and Abhishek Halder are with the Department of Aerospace Engineering, Iowa State University, Ames, IA 50011, USA, {\tt\footnotesize{\{rcsarker,ahalder\}@iastate.edu}}. 
}
\thanks{This research was partially supported by NSF awards CNS-2111688, CNS-1955670, and CCF-2326606.}
}

\IEEEoverridecommandlockouts
\begin{document}
\bstctlcite{IEEEexample:BSTcontrol}
\maketitle
\thispagestyle{empty}
\pagestyle{empty}

\begin{abstract}
We formulate a novel optimal control problem for data-driven dynamic resource allocation in multicore real-time systems. The proposed formulation accounts for hard deadline constraints and resource bounds. Under mild assumptions, we prove several results on the feasibility and existence of optimal solution for the proposed formulation. Building on these analyses, we design an algorithm to solve the multi-core dynamic resource allocation problem. We demonstrate the practical use of this new algorithm by experimentally  evaluating it on a real multicore platform using four benchmarks from the literature.
\end{abstract}
\section{Introduction}\label{sec:Introduction} 

\noindent\textbf{Motivation.} Real-time systems, such as those in automotive, robotic, and avionic applications, require predictable timing to operate safely. To meet growing computational demands, these systems are increasingly deployed on multicore platforms, where cores share hardware resources such as the last-level cache (LLC) and memory bus. Resource contention arising from this sharing can cause interference among concurrently executing programs, increasing their execution times and, in the worst case, causing deadline misses.

A common remedy is to \emph{statically} partition shared resources across cores~\cite{memguard, cat}, typically by solving a mixed-integer linear program (MILP) or an approximation thereof~\cite{xu2019holistic,sun2025multi}.
Because these allocations are static, they must be provisioned for the worst-case demand, even though a program's resource demand may vary over time. 
\emph{Dynamic allocation} addresses this limitation by redistributing shared-resource partitions across cores  at runtime in response to changing demand. Existing dynamic approaches \cite{singh2017survey}, however, rely on greedy heuristics~\cite{gifford2025rtas, eisenklam2025rasco}. While resource allocation for \emph{soft} deadline constraints has been studied before \cite{bondar2026generative}, to the best of our knowledge, no prior work computes \emph{optimal} dynamic allocations subject to \emph{hard} deadlines.

In this work, we formulate dynamic multicore resource allocation as a continuous-time continuous-state optimal control problem (OCP) subject to hard deadline constraints. 
 The objective of our OCP captures a system-level cost, such as energy consumption or provisioned resources, which is critical for embedded platforms with stringent size, weight, and power (SWaP) constraints. The continuous formulation pursued here is advantageous over discrete formulations in that it allows us to (i) circumvent computationally hard MILP formulations, and (ii) exploit vector field structures that discretize-then-optimize approaches cannot.
 

\noindent\textbf{Related work.}
Optimal control---both deterministic \cite{zhang2011model,ciocan2012model} and stochastic \cite{valian2016resource}---has been applied to general-purpose dynamic resource allocation. However, prior work does not address multicore real-time systems with hard deadlines. A separate line of work (see e.g., \cite{marti2004optimal,tabuada2007event}) considers resource allocation specifically for control tasks. Our focus, by contrast, is on general multicore tasks.

\noindent\textbf{Contributions.} This work makes the following specific contributions.
\begin{itemize}
    \item A bespoke OCP formulation of the dynamic multicore resource allocation with resource bounds and hard deadline constraints (Sec. \ref{sec:ModelAndProbFormulation}).

    \item Analysis of the proposed OCP (Sec. \ref{subsec:Analysis}), specifically its feasibility (Theorem \ref{thm:underapprox_feas}) and sufficient conditions for the existence of its solution (Theorem \ref{thm:suff_pw}, Proposition \ref{prop:obj_form}).

    \item A proposed algorithm (Algorithm \ref{alg:solution} in Sec. \ref{sec:Algorithm}) that builds on the analysis in Sec. \ref{subsec:Analysis}, and its evaluation (Sec. \ref{sec:Numerics}) using benchmarks (Table \ref{tab:benchmarks}) from the San Diego Vision Benchmark Suite~\cite{sd-vbs} on a real multicore platform.

\end{itemize}


\section{Model and Problem Formulation}\label{sec:ModelAndProbFormulation}
\noindent We consider a computer system with $s\in\mathbb{N}_{\geq 2}$ CPU cores, each executing a program (or task). We assume that each program executes a deterministic sequence of instructions. Let $\llbracket s\rrbracket:=\{1,\dots,s\}$ denote the task index set. Each task $i\in\llbracket s\rrbracket$ is released at a known time $t=t_{i}^{\text{start}}\geq 0$ and must complete by a known deadline $t=t_{i}^{\text{end}}$, where $t_{i}^{\text{start}}<t_{i}^{\text{end}}$. Assuming that the first task is released at $t=0$ and that all deadlines are finite, the system time window is the interval $$\left[\min_{i\in\llbracket s\rrbracket}t_{i}^{\text{start}},\max_{i\in\llbracket s\rrbracket}t_{i}^{\text{end}}\right] =: [0,T]\subset [0,\infty).$$
The system has $n_{r}\in\mathbb{N}$ shared resources that are to be split between all CPU cores in a manner such that all tasks meet their deadlines, their execution rates being dependent on their received allocation of the shared resources. From the physical configuration of the system, there is a total amount $R^{\rm{max}}_j$ of each resource $j\in\llbracket n_r\rrbracket$. The total system resources can then be represented by the vector $R^{\rm{max}}\in\R_+^{n_r}$.

We may thus represent the total allocation of all shared resources at time $t\in[0,T]$ as a matrix $R(t)\in\R_{+}^{s \times n_r}$. The $i$th row of this matrix, denoted as $R_i(t)\in\R_+^{n_r}$, is the total shared resource allocation to task $i$. Likewise, the $j$th column of this matrix, denoted as $R^j(t)\in\R_+^s$, is the allocation of resource $j$ among all tasks.

Note that if all of resource $j$ is used, then $R^j(t)/R^{\rm{max}}_{j}\in\Delta^{s-1}$, the standard simplex\footnote{For $s\in\mathbb{N}_{\geq 2}$, the standard simplex $\Delta^{s-1}$ is the convex hull of standard basis vectors in $\mathbb{R}^s$.}. Thus the normalized set of all valid allocations of resource $j$ is the convex hull $\rm{conv}(\mathbf{0},\Delta^{s-1})$.

\subsection{Resource-dependent Task Dynamics}\label{subsec:Taskdynamics}
\noindent For each task or program $i\in\llbracket s\rrbracket$, we refer to its execution progress, measured in total fraction of instructions retired at time $t\in[0,T]$ as $x_i(t)\in[0,1]$.

The corresponding normalized \emph{rate} of instructions retired (in instructions executed per second) is $\dot{x}_i(t)$. This latter quantity varies with both the state of program execution, represented by $x_i$, and the current resource allocation, given by $R_i$. For example, a program may be slow to execute in the beginning as data is loaded to memory, then faster once computation begins. Similarly, allocation of some resources may affect the execution rate far more so than other -- e.g., increasing memory bus bandwidth when the program is CPU-bound will have little to no effect. This dependence is expressed via the ordinary differential equation 
\begin{align}
    \dot{x}_i(t) = f_i(x_i(t),R_i(t)),
    \label{dynamicsTaski}
\end{align}
wherein for each task $i$, the vector field $f_i$ quantifies the rate of instructions retired over time as a function of (cumulative) current instructions retired $x_i(t)$, and current resource allocation $R_i(t)$. In other words, \eqref{dynamicsTaski} captures the \emph{dynamics} of task $i$.

In practice, the vector fields $f_i$ can be learned from empirical program execution data. However, the following properties can be assumed on $f_i$:
\begin{itemize}
    \item $f_i$ has no explicit dependence on other tasks' allocations, nor on the current time;
    \item $f_i$ is increasing in $R_i$, so giving a task more resources will not decrease its performance;
    \item $f_i$ is nonnegative.
\end{itemize}
We next outline the optimal control problem formulation.
 
\subsection{Optimal Control Formulation}\label{subsec:OCP}
\noindent We desire to find a time-varying allocation $R$ of all shared resources across given $s\in\mathbb{N}_{\geq 2}$ CPU cores to all tasks such that all tasks meet their deadlines -- we call such an allocation \emph{feasible}. Moreover, we desire to do so in some \emph{optimal} sense with respect to (w.r.t.) a performance objective $J$, which for an allocation $R$ has `cost' $J(R)$. Thus the problem of interest becomes

\begin{subequations}\label{eq:problem}
\begin{align}
    \underset{R:[0,T]\to\R_+^{s\times n_r}}{\text{minimize}} &\qquad J(R)\label{eq:problem_objective}\\
    \text{\quad subject to}\:\:\quad & \begin{aligned}
        &\dot{x}_i(t) = f_i(x_i(t),R_i(t)),\\
        & x_i(t) = 0 \quad \forall t\leq t_{i}^{\text{start}},\\
        & x_i(t) = 1 \quad \forall t\geq t_{i}^{\text{end}},\\
        & \bm{0} \preceq \sum_i R_i\preceq R^{\max},
    \end{aligned} \quad \forall i\in\llbracket s\rrbracket,\label{eq:problem_constraints}
\end{align}
\end{subequations}
where the constraints \eqref{eq:problem_constraints} fix the dynamics of each task, require that each task starts no earlier than its release time and ends no later than its deadline, and enforce resource limits. If these constraints are feasible, i.e. if at least one feasible allocation $R$ exists, an allocation which minimizes the objective \eqref{eq:problem_objective} is called \emph{optimal}.

\subsection{Performance Objective}\label{subsec:PerformanceObjective}
\noindent The objective $J$ is motivated by practical goals of system design and utilization, such as minimizing
\begin{itemize}
    \item the amount of resources required by the system while maintaining schedulability;
    \item cumulative use of a fixed set of resources;
    \item energy usage by a system.
\end{itemize}

\section{Analysis}\label{subsec:Analysis}
\noindent 

To begin our analysis for the OCP \eqref{eq:problem}, let us define the instantaneous set of admissible controls
\begin{equation}\label{eq:feasible_set_inst}
\mathcal{U}:=\left\{ R\in\R^{s\times n_r}_+ \mid \sum_iR_i\preceq R^{\rm{max}}\right\},
\end{equation}
and the total admissible control set
\begin{equation}\label{eq:feasible_set}
    \mathcal{F} := \left\{ R(\cdot) \mid \begin{aligned}
    R(t)\in\mathcal{U}&\quad\forall t\in[0,T] \\
    x_i(t^{\rm{end}}_i,R_i) \geq 1 &\quad\forall i\in\llbracket s\rrbracket
\end{aligned} \right\}.
\end{equation}
The inequality in \eqref{eq:feasible_set} is a slackening of the constraints \eqref{eq:problem_constraints} that will be useful in the development that follows.

\subsection{Separable Under-approximation}\label{subsec:Analysis_underapprox}
\noindent We consider the following assumption: for each $f_i,i\in\llbracket s\rrbracket$, there exists a nonnegative under-approximator $\underline{f_i}$, i.e., 
$$ 0\leq \underline{f_i}(x_i,R_i) \leq f_i(x_i,R_i),$$
which decomposes as
\begin{equation}\label{eq:assumption_underapprox}
    \underline{f_i}(x_i,R_i) := g_i(x_i)\cdot h_i(R_i)
\end{equation}
for some continuous and bounded functions
$$ g_i:[0,1]\to\mathbb{R}_{++} \quad,\quad h_i:\mathbb{R}_+^{n_r}\to\mathbb{R}_+\:. $$
Such an $\underline{f_i}$ is then a \emph{separable under-approximator} of the real task dynamics $f_i$. In words, $\underline{f_i}$ entails a lower rate of execution progress than $f_i$. We formalize this in the following Proposition \ref{prop:underapprox_feas}.
\begin{proposition}\label{prop:underapprox_feas}
    Let $\underline{\mathcal{F}}$ be the feasible control set of \eqref{eq:problem} under the dynamics $\underline{f_i}$. Then, any solution of \eqref{eq:problem} under the dynamics $\underline{f_i}$ is also feasible under the dynamics $f_i$. I.e, $$ \underline{\mathcal{F}} \subset \mathcal{F} \:. $$
\end{proposition}
\begin{proof}
    Let $R\in\underline{\mathcal{F}}$. Since $\underline{f_i}\leq f_i$ pointwise, and $\underline{{x}_i}(t_i^{\text{start}}) = {x}_i(t_i^{\text{start}}) = 0$, by the comparison lemma \cite[Lemma 3.4]{khalil2002nonlinear}, we have
    $$ \underline{x_i}(t) \leq x_i(t) \quad\forall t\in[t_i^{\text{start}},t_i^{\text{end}}]\:.$$
    Therefore,
    $$ \underline{x_i}(t_i^{\text{end}})\geq 1 \quad\Longrightarrow\quad {x_i}(t_i^{\text{end}})\geq 1\:, $$
    which shows that $R\in\mathcal{F}$, and the statement follows.
\end{proof}
We next show that the separability of $\underline{f_i}$ allows us to express feasibility in a different way.
\begin{lemma}\label{lemma:transformed_feas}
    Let  
    $$ G_i(x) :=\int_0^x\frac{d\xi}{g_i(\xi)}\quad\text{and}\quad W_i:= G_i(1)\quad\forall i\in\llbracket s\rrbracket. $$
    Then, 
    $$ R\in\underline{\mathcal{F}} \quad\Longleftrightarrow\quad \int_{t_i^{\mathrm{start}}}^{t_i^{\mathrm{end}}}h_i(R_i(t))dt \geq W_i \quad\forall i\in\llbracket s\rrbracket.$$
\end{lemma}
\begin{proof}
    Define the transformed state
    $$z_i(t) := G_i(\underline{x_i}(t)),$$
    and its dynamics
    $$\dot{z}_i = G'_i(\underline{x_i})\dot{x}_i = \frac{1}{g_i(\underline{x_i})} g_i(\underline{x_i})h_i(R_i) = h_i(R_i).$$
Since $g_i>0$, $G_i$ is strictly increasing, and so
    $$ \underline{x_i}(t_i^{\text{end}}) \geq 1 \quad\Longleftrightarrow\quad z_i(t_i^{\text{end}}) = G_i(\underline{x_i}(t_i^{\text{end}})) \geq G_i(1) = W_i\:. $$
    As $z_i(t_i^{\text{start}})=0$, 
    $$ z_i(t_i^{\text{end}}) = \int_{t_i^{\text{start}}}^{t_i^{\text{end}}}h_i(R_i(t))dt\:, $$
    completing the proof.
\end{proof}
Additionally, when $h_i$ are concave, then we can deduce the following result in Proposition \ref{prop:feas_piecewise} on the nature of feasible solutions. In practice, this concavity assumption holds naturally due to diminishing returns -- generally the more resources a task has, the less benefit there is to allocating yet more.
\begin{proposition}\label{prop:feas_piecewise}
    Suppose that some $R\in\underline{\mathcal{F}}$ exists, and that $h_i$ in \eqref{eq:assumption_underapprox} are concave for all $i\in\llbracket s\rrbracket$. Then there exists a piecewise-constant $\bar{R}\in\underline{\mathcal{F}}$ which changes only at release/deadline times.
\end{proposition}
\begin{proof}
    Let us sort the distinct values $0,T$ along with $\{t_i^{\text{start}},t_i^{\text{end}}\}_{i\in\llbracket s\rrbracket}$ as
    $$ 0=\tau_0<\tau_1<\dots<\tau_{K}=T\:.$$
    We then define the resulting intervals and lengths
    $$ I_k:=[\tau_{k-1},\tau_k)\:,\quad\delta_k:=\tau_k-\tau_{k-1} \quad\forall k\in\llbracket K\rrbracket\:,$$
    as well as the set of tasks running on each interval
    $$ A_k := \{i \mid I_k\subset[t_i^{\text{start}},t_i^{\text{end}}]\}\:. $$
    Now define $\bar{R}$ piecewise as 
    $$ \bar{R}_i(t) = \bar{r}_{ik}\:,\quad\text{where}\quad\bar{r}_{ik} := \frac{1}{\delta_k}\int_{I_k} R_i(t)dt\:  $$
    is the average resource allocation of task $i\in A_k$ over the interval $I_k$. To see that $\bar{R}(t)\in\mathcal{U}$, observe that
    $$ \sum_{i} \bar{R}_i(t)=\sum_{i}\bar{r}_{ik} = \frac{1}{\delta_k}\int_{I_k}\sum_iR_i(t)dt\leq R^{\max}\:, $$
    since $R\in\mathcal{U}$. 
    
    Now, to see that $\bar{R}$ is feasible, we apply Jensen's inequality to $h_i$ concave, to find
    $$ h_i(\bar{r}_{ik}) = h_i\left(\frac{1}{\delta_k}\int_{I_k} R_i(t)dt\right) \geq \frac{1}{\delta_k}\int_{I_k} h_i(R_i(t))dt\:, $$
    and thus
    $$ \delta_k h_i(\bar{r}_{ik}) \geq \int_{I_k} h_i(R_i(t))dt\:. $$
    Since $[t_i^{\text{start}},t_i^{\text{end}}]$ can be partitioned into the intervals $\{I_k\mid i\in A_k\}$, we have
    \begin{align*}
        \int_{t_i^{\text{start}}}^{t_i^{\text{end}}} h_i(\bar{R}_i(t))dt &= \sum_{\{k\mid i\in A_k\}}\delta_k h_i(\bar{r}_{ik}) \\
        &\geq \int_{t_i^{\text{start}}}^{t_i^{\text{end}}} h_i(R_i(t))dt  \geq W_i \:,
    \end{align*}
    where the last inequality follows from feasibility of $R$ and Lemma \ref{lemma:transformed_feas}, whereby the desired result follows.
\end{proof}

Under the stated assumptions, Propositions \ref{prop:underapprox_feas} and \ref{prop:feas_piecewise} together tell us that to ascertain whether a feasible $R\in\mathcal{F}$ exists, it is sufficient to verify that a piecewise-constant $\bar{R}\in\underline{\mathcal{F}}$ as in Proposition \ref{prop:feas_piecewise} exists. This, in turn, requires us to solve the feasibility problem:
\begin{subequations}\label{eq:problem_underfeas}
\begin{align}
    \underset{\{r_{ik}\in\mathbb{R}_+^{n_r}\}_{i\in\llbracket s\rrbracket}^{k\in\llbracket K\rrbracket}}{\text{minimize}} &\qquad 0\label{eq:problem_underfeas_objective}\\
    \text{\quad subject to}\:\:\quad & \begin{aligned}
        & \sum_{\{k\mid i\in A_k\}}\delta_k h_i({r}_{ik}) \geq W_i,\\
        & \bm{0} \preceq r_{ik}\:,\: \sum_i r_{ik}\preceq R^{\max},
    \end{aligned} \quad \forall i\in\llbracket s\rrbracket.\label{eq:problem_underfeas_constraints}
\end{align}
\end{subequations}
In fact, letting $\underline{\mathcal{F}_{\rm{pw}}}$ be the set of piecewise-constant allocations satisfying \eqref{eq:problem_underfeas_constraints}, Corollary \ref{cor:pw_alloc_suff} next shows that the non-emptiness of this set is \emph{equivalent} to that of $\underline{\mathcal{F}}$. 

\begin{corollary}\label{cor:pw_alloc_suff}
$$ \underline{\mathcal{F}_{\rm{pw}}}\neq\emptyset\quad\Longleftrightarrow\quad\underline{\mathcal{F}}\neq\emptyset.$$
\end{corollary}
\begin{proof}
The direction $(\Rightarrow)$ follows since members of $\underline{\mathcal{F}_{\rm{pw}}}$ define a feasible control under the dynamics $\underline{f_i}$. The direction $(\Leftarrow)$ was proven in Proposition \ref{prop:feas_piecewise}.
\end{proof}
Furthermore, we have the following result.
\begin{lemma}\label{lemma:pw_alloc_convex}
$\underline{\mathcal{F}_{\rm{pw}}}$, the constraint set of \eqref{eq:problem_underfeas}, is bounded and convex. If all $h_i$ are concave as before, then $\underline{\mathcal{F}_{\rm{pw}}}$ is also closed ($\Rightarrow$ compact).
\end{lemma}
\begin{proof}
    First, notice that the set $\underline{\mathcal{F}_{\rm{pw}}}$ in \eqref{eq:problem_underfeas_constraints} is finite dimensional. The constraints $0\preceq r_{ik}$ and $\sum_i r_{ik}\preceq R^{\max}$ are closed, bounding and affine in $\{r_{ik}\}$. 
    
    The function $\sum_{\{k\mid i\in A_k\}}\delta_k h_i({r}_{ik})$, being a nonnegative sum of concave functions (defined on the interior of the nonnegative orthant) of $\{r_{ik}\}$, is also concave, and thus continuous (since their domains are open convex sets). Its superlevel sets, which make up the remaining constraint, are therefore also convex and closed. Thus the intersection of these constraining sets, namely $\underline{\mathcal{F}_{\rm{pw}}}$, is itself convex and closed.
\end{proof}

We summarize the developments above in the following Theorem \ref{thm:underapprox_feas}.
\begin{theorem}\label{thm:underapprox_feas}
    For each $i\in\llbracket s\rrbracket$, if there exists a separable under-approximation
    $$ 0\leq \underline{f_i}(x_i,R_i) = g_i(x_i)h_i(R_i) \leq f_i(x_i,R_i) $$
    where $g_i>0$ and $h_i\geq 0$ concave, then the \emph{convex feasibility problem} \eqref{eq:problem_underfeas} guarantees the existence of a feasible allocation $R\in\mathcal{F}$. In other words,
    $$ \eqref{eq:problem_underfeas}\text{ feasible }\quad\Longleftrightarrow\quad \underline{\mathcal{F}}\neq\emptyset\quad\Longrightarrow\quad\mathcal{F}\neq\emptyset. $$
    When the upder-approximation is lossless ($\underline{f_i}=f_i$), the right arrow above becomes bidirectional.
\end{theorem}

\subsection{Existence}\label{subsec:Analysis_existence}
\noindent Theorem \ref{thm:underapprox_feas} says that if a piecewise-constant allocation exists for the underestimated task dynamics $\underline{f_i}$, then this same allocation will be feasible for the real dynamics $f_i$. This is not enough, however, to prove the existence of an \emph{optimal} allocation, i.e. a solution to the OCP \eqref{eq:problem}, as there is no guarantee that the minimizer of objective $J$ over $\mathcal{F}$ lies in $\underline{\mathcal{F}_{\rm{pw}}}$. We now turn to identify conditions where this is indeed the case.

The natural condition therefor is that $J$ cannot increase with intra-event averaging. This is shown in Theorem \ref{thm:suff_pw} next. For this purpose, we define the averaging operator $\mathcal{A}$ from the proof of Proposition \ref{prop:feas_piecewise} as
$$ (\mathcal{A}[R])_i(t) := \overline{r}_{ik} = \frac{1}{\delta_k}\int_{I_k}R_i(s)ds \quad t\in I_k\:,$$
to succinctly write this condition as
\begin{equation}\label{eq:cond_suff_pw}
J(\mathcal{A}[R]) \leq J(R) \:.
\end{equation}
Thereunder, the minimal value of the objective \eqref{eq:problem_objective} must be attained by a piecewise-constant allocation.

\begin{theorem}\label{thm:suff_pw}
    Assume that condition \eqref{eq:cond_suff_pw} holds on $J$ for all $R\in\underline{\mathcal{F}}$. Then,
    \begin{equation}
        \inf_{R\in\underline{\mathcal{F}}} J(R) = \inf_{R\in\underline{\mathcal{F}_{\rm{pw}}}} J(R)\:.
    \end{equation}
    If additionally $J$ is lower semicontinuous over $\underline{\mathcal{F}_{\rm{pw}}}$, then this infimum is attained by a minimizer in $\underline{\mathcal{F}_{\rm{pw}}}$.
\end{theorem}
\begin{proof}
The inequality $(\leq)$ holds trivially as $\underline{\mathcal{F}}\supset\underline{\mathcal{F}_{\mathrm{pw}}}$. The reverse holds by condition \eqref{eq:cond_suff_pw}, i.e., $\forall R\in\underline{\mathcal{F}}$, there exists $\mathcal{A}[R]\in\underline{\mathcal{F}_{\rm{pw}}}$ with lower cost in $J$.

Recall that the set $\underline{\mathcal{F}_{\rm{pw}}}$ is finite-dimensional and is compact (by Lemma \ref{lemma:pw_alloc_convex}). By Weierstrass extreme value theorem \cite[p. 89, Theorem 4.16]{rudin1976principles}, $J$ attains its minimum on $\underline{\mathcal{F}_{\rm{pw}}}$.
\end{proof}

Condition \eqref{eq:cond_suff_pw} can be shown to hold for a breadth of costs, such as those penalizing sustained and peak resource use. This is captured in Proposition \ref{prop:obj_form}.

\begin{proposition}\label{prop:obj_form}
    If objective function $J$ takes either the form
    $$ J_1(R) = \int_0^T\mathscr{L}(R(t))dt \quad\text{or}\quad J_2(R)=\max_{t\in[0,T]}\phi(R(t))$$
    with $\mathscr{L}$ and $\phi$ convex, then the condition \eqref{eq:cond_suff_pw} is satisfied.
\end{proposition}
\begin{proof}
    Observe that
    \begin{align*}
        J_1(\mathcal{A}[R]) &= \sum_{k}\int_{I_k}\mathscr{L}(\overline{r}_{k})dt = \sum_k\delta_k\mathscr{L}(\overline{r}_{k}) \\
        &\leq \sum_k\int_{I_k}\mathscr{L}(R(t))dt = J_1(R),
    \end{align*}
    where the inequality above follows from convexity of $\mathscr{L}$ and Jensen's inequality. For $J_2$, we have (again by Jensen)
    $$ \phi(\overline{r}_k) \leq \frac{1}{\delta_k}\int_{I_k}\phi(R(t))dt \leq \max_{t\in I_k}\phi(R(t))\:,$$
    and so $J_2(\mathcal{A}[R])\leq J_2(R)$, as desired. 
\end{proof}

\begin{remark}\label{rmk:J_convexity_vs_pw}
The costs of the forms in Proposition \ref{prop:obj_form} are both convex in $R$, which coupled with the convex set $\underline{\mathcal{F_{\rm{pw}}}}$ (Lemma \ref{lemma:pw_alloc_convex}) makes \eqref{eq:problem} a finite-dimensional convex program, which can be solved via off-the-shelf algorithms. Moreover, we note that neither condition \eqref{eq:cond_suff_pw} nor convexity subsumes the other -- for a candidate cost $J$, both need to be verified.
\end{remark}

\section{Proposed Algorithm}\label{sec:Algorithm}
\noindent Under the assumption that our resource cost $J$ is both convex \emph{and} satisfies \eqref{eq:cond_suff_pw}, we now present our approach to the solution of problem \eqref{eq:problem} as Algorithm \ref{alg:solution} which takes us from empirical data on task resource usage to a locally optimal allocation on a given taskset. 

The proposed Algorithm \ref{alg:solution} first fits true and under-approximated dynamics to each task, then finds the optimal allocation for the underestimated dynamics (a finite-dimensional, convex program with the assumption on $J$, and yielding a piecewise-constant allocation as the optimum, by Theorem \ref{thm:suff_pw}). It then uses this allocation as a starting point for solving the nonlinear, nonconvex program \eqref{eq:problem} with the real dynamics. In the following, we explain the steps in the proposed Algorithm \ref{alg:solution}.

\begin{algorithm}[t!]
\caption{Solving Optimal Resource Allocation Problem}
\begin{algorithmic}[1]
    \Require Profiles $\{\dot{x}_i(\tau)\}_{r\in R^{\text{known}}}^{i\in\llbracket s\rrbracket}$ for all $s$ tasks and for all known contexts $R^{\text{known}}$; resource limits $R^{\max}$;  release/deadline times $\{t_i^{\text{start}},t_i^{\text{end}}\}_{i\in\llbracket s\rrbracket}$; function fitting algorithms \texttt{fit\_dynamics} and \texttt{underapprox}; convex cost functional $J$ satisfying \eqref{eq:cond_suff_pw}; convex and nonconvex program routines \texttt{solve\_convex}, \texttt{solve\_nonconvex}.\newline

    \Statex \textit{Step 0 (dynamics fitting and under-approximation)):}
    \For{$i\in\llbracket s\rrbracket$}
        \State $f_i \gets \texttt{fit\_dynamics}(\{\dot{x}_i(\tau)\}_{r\in R^{\text{known}}}^{i\in\llbracket s\rrbracket})$
        \State $g_i, h_i \gets \texttt{underapprox}(f_i)$
        \State $\underline{f_i}(x,R) \gets g_i(x)\,h_i(R)$
        \State $W_i \gets \displaystyle\int_0^1 \frac{dx}{g_i(x)}$ 
    \EndFor\newline

    \Statex \textit{Step 1 (convex relaxation, Thm.~1 \& Prop.~3):}
    \State $\{r_{ik}\} \gets \texttt{solve\_convex}\big(J,\eqref{eq:problem_underfeas_constraints}\big)$ 
    \newline

    \Statex \textit{Step 2 (warm-started solve under the true dynamics):}
    \State $R_{\mathrm{opt}} \gets \texttt{solve\_nonconvex}\big(J,\{r_{ik}\},\eqref{eq:problem_constraints}\big)$ 
    \newline
\end{algorithmic}
{\bf{Result:}} Optimal allocations $\{r_{ik}\}$ and $R_{\mathrm{opt}}$ for $\eqref{eq:problem}$ with cost $J$ and dynamics $\{\underline{f_i}\}$ and $\{f_i\}$ respectively.
\label{alg:solution}
\end{algorithm}

\subsection*{Step 0: Dynamics fitting and under-approximation}
\noindent  For each task, we begin by fitting continuous dynamics to our empirical profiling data through the \texttt{fit\_dynamics} routine. In our implementation, we perform isotonic piecewise-linear regression with a floor constraint of 0, ensuring that the resulting $f_i$ is monotone in $R$ and nonnegative. In practice a different fit may be used, so long as those conditions are enforced.

The \texttt{underapprox} routine takes $g_i(\cdot):=f_i(\cdot,R^{\max})$, and then finds $h_i$ as the concave minorant of $\min_x\frac{f_i(x,\cdot)}{g_i(x)}$. We then explicitly enforce positivity and nonnegativity of these functions respectively via a floor constraint, although we find in practice this is not needed. The under-approximator dynamics $\underline{f_i}$ and the useful quantity $W_i$ from Lemma \ref{lemma:transformed_feas} are then computed as part of this step as well.

\subsection*{Step 1: Solving the under-approximated problem}
\noindent The constraints \eqref{eq:problem_underfeas_constraints}, obtained through the computed $\{h_i\}$, $\{W_i\}$, release/deadline times, and $R^{\max}$, are fed along with objective $J$ into a finite-dimensional convex optimization solver (\texttt{scipy}'s \texttt{minimize}, in our implementation) to find the optimal resource allocation for the taskset under the under-approximated dynamics $\{\underline{f_i}\}$ -- by Theorem \ref{thm:suff_pw}, this will be the piecewise-constant allocation $\{r_{ik}\}$ which changes only at task release/deadline events.

\subsection*{Step 2: Solving \eqref{eq:problem} with the real dynamics}
\noindent  The solution $\{r_{ik}\}$ of the previous step is important for this step as well -- it proves that a feasible allocation exists (Theorem \ref{thm:underapprox_feas}), and provides a starting point for the solver of \eqref{eq:problem} -- \texttt{solve\_nonconvex}. In our implementation, we used the sequential least-squares quadratic programming (SLSQP) solver provided in \texttt{scipy}, as we found it to be suitable for our constraint set and problem dimension. Starting from $\{r_{ik}\}$, the solver iteratively converges to a locally cost-minimizing allocation, which we call $R_{\mathrm{opt}}$. Although we assume $J$ is convex, the real dynamics $\{f_i\}$ need not form a convex constraint set \eqref{eq:problem_constraints}, so $R_{\textrm{opt}}$ need not be globally optimal. 


\section{Evaluation}\label{sec:Numerics}
\noindent To evaluate the effectiveness of the proposed approach in terms of resource efficiency and safety, measured as objective minimization \eqref{eq:problem_objective} and deadline satisfaction \eqref{eq:problem_constraints}, we conducted a case study on a real multicore platform.

\subsection{Evaluation Setup}
\noindent\textbf{Programs.} Our case study uses benchmarks from the San Diego Vision Benchmark Suite (SD-VBS)~\cite{sd-vbs}, following prior work on multicore resource allocation~\cite{sun2025multi}.
Table~\ref{tab:benchmarks} summarizes the benchmarks. Each benchmark was profiled and executed using the \texttt{vga} input class from SD-VBS.

\noindent\textbf{Platform.}  We performed all measurements and ran all experiments on an Intel Xeon E5-2618L v3 processor with 8 cores, a 20 MB 20-way set associative last-level L3 cache, and a single-channel 8 GB PC-2133 DDR4 DRAM. We consider two shared resources: the LLC and memory bandwidth, each divided into 20 partitions using Intel's Cache Allocation Technology~\cite{cat} and  MemGuard~\cite{memguard}, respectively. Because the Cache Allocation Technology enforces a minimum allocation of 2 partitions per core  on our platform, the set of feasible allocations, prior to normalization, is$$R^{\text{known}} \coloneqq \{2, 3,  \dots, 20\} \times \{1, 2, \dots, 20\}.$$

\noindent\textbf{Profiling.} For each benchmark $i$ in Table~\ref{tab:benchmarks}, we learn its dynamics $f_i$ from empirical profiling data. To collect these data, we pin the benchmark to a single core, assign it a feasible resource vector $r$, and sample its state $x_i(t)$ every 10\,ms. We repeat this procedure five times for each possible resource allocation $r$, and then {compute the corresponding state derivatives $\dot{x}_i(t)$ from the collected time-series data. This yields a dataset of input-output pairs $\big((x_i(t), r),\, \dot{x}_i(t)\big)$ collected over all $r \in R^{\text{known}}$, from which we learn $f_i$.}

\begin{table}[t]
  \centering
  \caption{SD-VBS Benchmarks~\cite{sd-vbs} Used in Sec. \ref{sec:Numerics}}
  \label{tab:benchmarks}
  \begin{tabularx}{\columnwidth}{@{} l X l @{}}
    \toprule
    \textbf{Benchmark} & \textbf{Description} & \textbf{Characteristic} \\
    \midrule
    \texttt{disparity} & Depth perception for  cruise / collision control, tracking  & Data-intensive \\
    
    \texttt{localization} & Estimation  of position in environment from sensor data & Compute-intensive \\
    
    \texttt{sift} & Scale-invariant feature transform for object recognition & Compute-intensive \\
    \texttt{tracking} & Motion extraction from image sequences & Data-intensive \\
    \bottomrule
  \end{tabularx}
\end{table}

\begin{figure*}[htbp]
    \centering
\includegraphics[width=0.99\linewidth]{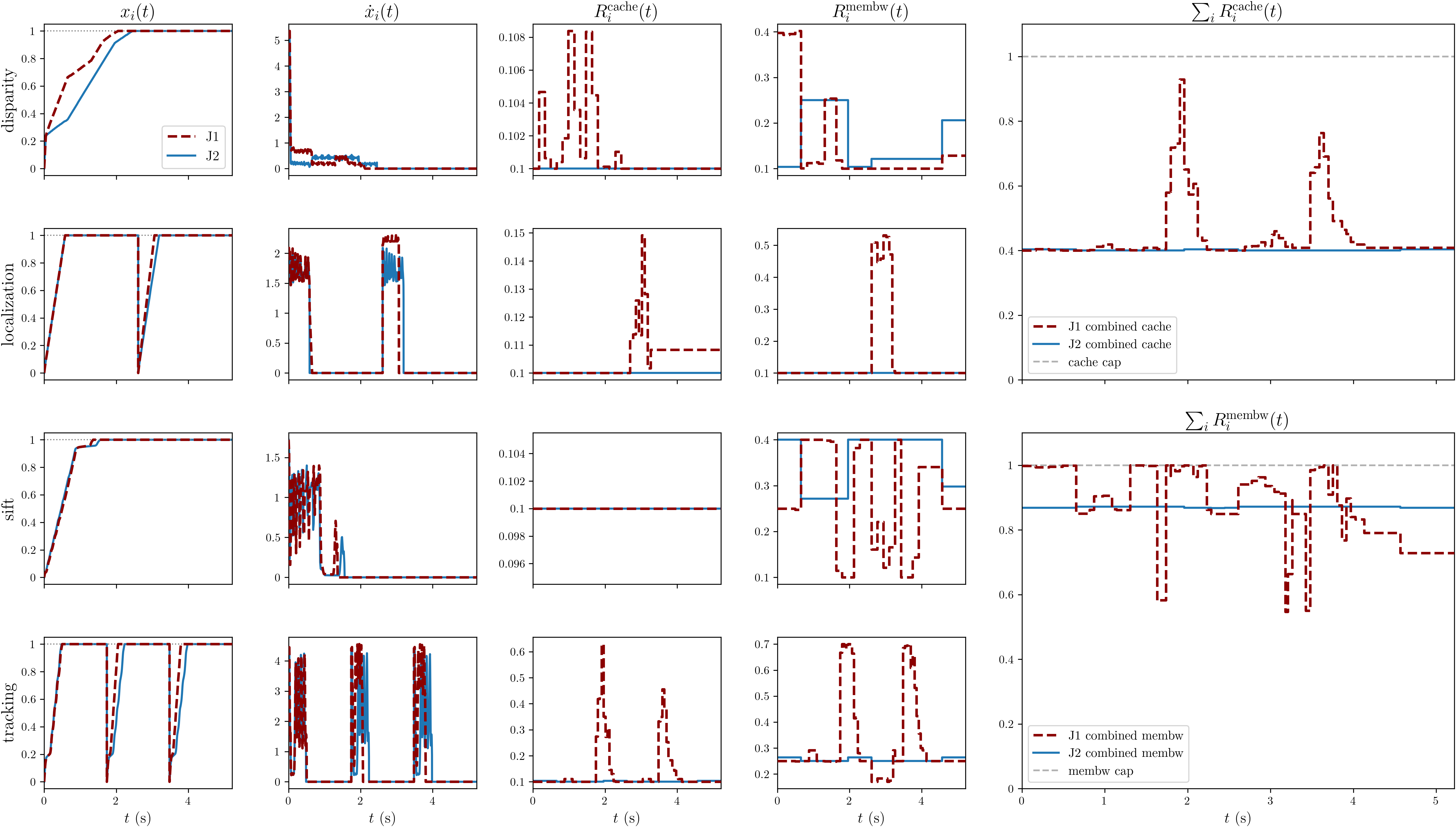}
\caption{Numerical evaluation results for Sec. \ref{sec:Numerics} showing locally optimal resource allocations and resulting execution profiles for the evaluated task set with costs $J_1$ ({{red, dashed}}) and $J_2$ ({{blue}}) as in \eqref{eq:J1_eval} and \eqref{eq:J2_eval} respectively. We can see the former front-loads resources to maximize `dead time', whereas the latter prioritizes a near-constant, low-peak allocation.
}
\label{fig:EndpointPDFs}
\end{figure*}

\noindent\textbf{Task set.} 
The goal of our case study is to compute a resource-optimal schedule for a real-time task set built from the profiled SD-VBS benchmarks, where each task is assigned to its own core. We consider periodic tasks with implicit deadlines and synchronous release; that is, all tasks are released simultaneously, and each task must complete by the start of its next period. Based on the profiled execution times, we set the task periods so that the total {\em reference} utilization is exactly 1.0, where a task's reference utilization is the ratio of its execution time under an even partitioning of shared resources to its period. This yields a system window, defined as the least common multiple of the task periods, of 5.22 seconds. Within each system window, the task set releases 1, 2, 1, and 3 instances of \texttt{disparity}, \texttt{localization}, \texttt{sift}, and \texttt{tracking}, respectively; the corresponding periods are 5.22, 2.61, 5.22, and 1.74 seconds.

\subsection{Numerical Evaluation Results}


\noindent\textbf{Objectives.}
We compute resource-optimal schedules for two objective functions, below:
\begin{align}
    J_1(R) &:= \int_0^T \chi\Big[\, \exists\, i : t\geq t_i^{\rm{start}}\land x_i(t)<1 \,\Big]\, dt, \label{eq:J1_eval}\\
    J_2(R) &:= \max_{t\in[0,T]}\sum_i\bm{1}^\top R_i(t),\label{eq:J2_eval}
\end{align}
where $\chi\left[\cdot\right]$ in \eqref{eq:J1_eval} is the indicator function, in this case equal to 1 if there is at least one task which has been released but has not completed, else 0. I.e., minimizing $J_1$ is equivalent to maximizing the time during which all cores are idle.
This objective encourages the cores to ``race to idle,'' a common energy-aware scheduling strategy that allows the CPU package to enter a low-power sleep state.

Minimizing the second cost, $J_2$, penalizes spikes in resource usage and thus should lead to a shallow, low allocation. In \eqref{eq:J2_eval}, $\bm{1}$ denotes an all-ones vector. Rather than reducing energy consumption on a fixed platform, $J_2$ identifies the minimal hardware resources that guarantee all deadlines are met, allowing designers to provision lower-cost hardware.

Notably, it is easy to see from \eqref{eq:J1_eval} and \eqref{eq:J2_eval} that $J_2$ is of the form $J_2$ in Proposition \ref{prop:obj_form} -- both convex in $R$ and satisfying \eqref{eq:cond_suff_pw} -- and similarly that $J_1$ has neither of these properties. The choice of these two costs in particular is intentional -- $J_2$ is nicely behaved and satisfies our assumptions in this work, whereas $J_1$ is not and does not. We expect to see the latter potentially benefit much more from a continuous allocation than from a piecewise-constant one. Both costs represent real-world design goals of interest to the real-time community.

\noindent\textbf{Results.} We ran Algorithm \ref{alg:solution} for the same taskset with cost $J_2$ as in \eqref{eq:J2_eval}. For $J_1$ as in \eqref{eq:J1_eval} we run a modified version of Step 2 of the algorithm, warm-starting from the feasible piecewise-constant allocation of cost $J_1$. 

The resulting resource-optimal schedules computed for both $J_1$ and $J_2$ are visualized in Fig.~\ref{fig:EndpointPDFs}. We can see that, as one should expect, the cost imposed affects the nature of the optimal allocation -- the $J_1$ solution spikes resource allocation to tasks as soon as they are released in order to quickly complete them and maximize time that the system is idle, whereas $J_2$'s solution prioritizes constant total resource allocation, thereby minimizing system resource requirements. We also observe that, for this taskset, memory bandwidth is the more impactful resource, with both allocations varying it with greater magnitude to control task execution rate -- indeed, $J_2$'s solution keeps the cache allocation almost entirely constant for all tasks throughout the hyperperiod $[0,T]$.

On a Debian 12 machine with an AMD Ryzen 7 5800X CPU and 80GB of RAM, the runtimes for the full algorithm under costs $J_1$ and $J_2$ were $\approx5$ minutes and $\approx 14$ minutes respectively. As expected, the better-structured cost $J_2$ was quicker to solve for.

\subsection{Empirical Evaluation Results}
\noindent We next validate on our hardware platform that the computed schedule preserves schedulability in practice and compare its energy consumption against a static baseline. We consider a baseline allocation that statically allocates the full resource capacity of each shared resource evenly across cores. We execute the $J_1$ and $J_2$-optimal schedules and the baseline on our hardware platform, and measure their energy consumption and schedulability over the system window.

\noindent\textbf{Results.} The results show that all tasks indeed meet their deadlines when executing on the real platform under both the $J_1$ and $J_2$ schedules. The $J_2$-optimal schedule reduces the required number of cache and memory bandwidth partitions by 60\% and 15\%, respectively, relative to the full capacity of 20 partitions per resource. Further, the $J_1$-optimal schedule reduces the controllable energy consumption (energy consumed beyond the static floor) by 6.3\% compared to the baseline. These results demonstrate that our OCP formulation and solver can meaningfully optimize system-level objectives via dynamic resource allocation while guaranteeing timing predictability.

\section{Concluding Remarks}\label{sec:Conclusion} 
\noindent We presented an optimal control approach for dynamic shared-resource allocation in multicore real-time systems with hard deadlines. By modeling task progress as controlled dynamics and optimizing time-varying resource allocations, our approach reduces both provisioned resource capacity and energy consumption while guaranteeing deadline feasibility. Future work will extend the approach to multiple tasks per core and incorporate online adaptation to account for potential discrepancies between profiled and actual task dynamics at runtime.



\balance

\bibliographystyle{IEEEtran}
\bibliography{myreferences}

\end{document}